\documentclass[graybox]{svmult}

\usepackage{type1cm}        
\usepackage{makeidx}         
\usepackage{graphicx}        
\usepackage{multicol}        
\usepackage[bottom]{footmisc}

\usepackage{newtxtext}       %
\usepackage{newtxmath}       
\usepackage{hyperref}
\usepackage{fixme}

\newcommand{\meaningof}[1]{\llbracket #1 \rrbracket}

{

\makeindex             

\begin{document}

\title*{Analysis and Design of Uncertain Cyber-Physical Systems}
\author{Alessandro Pinto
}
\institute{Alessandro Pinto \at Work performed while at Raytheon Technologies Research Center, 2855 Telegraph Avenue, Berkeley, CA (current contact: apinto@jpl.nasa.gov; web: \url{http://www.alessandropinto.net} )
}
%
%
\maketitle

\abstract*{Current design practices for Cyber-Physical Systems (CPS) leverage several methods to deal with uncertainty in the requirements, environment, and implementation platform, such as physical and functional redundancy. CPS have evolved in the past decades towards a higher-level of autonomy, and a more decentralized and connected implementation. The evolution towards more autonomous systems has changed the computation and communication workloads, demanding particular care in the early design phases to avoid exceeding typical size, power, and cost constraints. Moreover, the use of approximate models, the complexity of the state estimation and controlled problems, and imperfect  communications, suggest that epistemic uncertainty will play a major role in these systems. 
After presenting the evolution of CPS over the past two decades, we review the main sources of uncertainty in classical CPS with emphasis on the implementation platform such as failures, timing, and implementation bugs. We present several classical methods to deal with uncertainty, and explain why these methods, while still applicable to autonomous systems, are not sufficient. Finally, we present a compositional framework that focuses on requirements, and that supports reasoning about aleatoric and epistemic uncertainty.}

\begin{description}
    \item[CPS]{Cyber-Physical System}    
    \item[AI]{Artificial Intelligence}
    \item[MQTT]{Message Queuing Telemetry Transport}
    \item[UAV]{Unmanned Aerial Vehicle}
    \item[MTTF]{Mean Time To Failure}  
    \item[FMEA]{Failure Mode and Effect Analysis}
    \item[FTA]{Fault Tree Analysis}
    \item[WCET]{Worst Case Execution Time}
    \item[EDF]{Earliest Deadline First}  
    \item[RTOS]{Real Time Operating System}
    \item[TTP]{Time Triggered Protocol}
    \item[HOL]{Higher Order Logic}
    \item[JADC2]{Joint All-Domain Command and Control}   
    \item[LiDAR]{Light Detection and Ranging}     
    \item[GNSS]{ Global Navigation Satellite System} 
    \item[NASA]{National Aeronautics and Space Administration}  
\end{description}

Several sources of uncertainty have to be taken into account in the analysis and design of CPS. The set of parameters used in the model of the physical plant may be uncertain due, for example, to manufacturing processes that are precise up to some bounded tolerance. Physical quantities are sensed by electronic components that add noise to the sensed signals. Abstraction of the physical world, which is often necessary to limit the complexity of the models used in analysis and at run-time in decision-making, introduces inaccuracies. The cyber side of a CPS, which includes both hardware and software components, exposes several types of uncertainty such as failures, latency, and implementation errors.

Design processes and tools allow engineers to minimize the impact of these types of uncertainty, and to deliver systems which can be operated with an acceptable level of risk. Real-time operating systems can provide hard guarantees on the execution of control functions, and several methods can be used to compute bounds on the worst-case execution time of software. Similarly, time-triggered or priorities-based communication protocols can provide hard bounds on the delay experienced by messages exchanged over a network. Failures can be handled through redundant architectures that are typically used in safety critical systems, while rigorous software and hardware development processes can reduce the chances of implementation errors. Protection from cyber-attacks can be achieved through physical isolation, careful control of the supply-chain, and other strategies such as authentication and encryption at the boundary of the system. Furthermore, in some cases, cyber-physical systems can be modelled with enough fidelity to enable the effective use of formal analysis and synthesis which eliminates human errors. In general, however, delivering reliable systems in the face of uncertainty comes at a high price, which demands for a careful evaluation of risks.  

In the past several years, cyber-physical systems have evolved, primarily due to pervasive connectivity, miniaturization, cost-effectiveness of hardware, and advances in the area of Artificial Intelligence (AI). The complexity of the cyber side of the CPS is constantly increasing as more functions are moved to software. Correspondingly, the hardware platforms on which this software runs have become more complex. Another evolution has occurred in the area of communications at all levels of the protocol stack allowing disparate computing platforms to be interconnected. Connectivity is pervasive today leveraging technologies such as WiFi, 4G and 5G, and higher-level protocols for the internet of things such as the Message Queuing Telemetry Transport (MQTT). The ability to connect sensing and actuation nodes to the internet enables deployment over a cloud computing infrastructure where data can be aggregated, and computationally expensive algorithms can be executed with virtually unlimited resources. This trend has driven the development of new applications that control societal-scale systems such as the electric grid, transportation, and logistics. Coupling between the physical world and the cyber world can also occur in unconventional ways such as through social media platforms that shape consumer preferences and even opinions of people, driving their physical acts. 

These new class of applications features an environment which is much more complex to model than traditional physical systems due not only to their scale, but also to new sources and types of uncertainty. Consider, for example, the case of echo chambers which seems to result from algorithms trying to learn user interests and preferences. These types of effects are not easily predictable because of high uncertainty in the environment (people in this case), which is only approximately represented by machine learning models, but that is inherently due to \emph{lack of knowledge}. New models and analysis methods are therefore needed to capture different types of uncertainties, and to analyze these new classes of systems. 

In this chapter, we start by discussing how cyber-physical systems have evolved from simple controllers to networks of highly autonomous systems. In Section \ref{sec:uncertainty}, we describe some of the most common sources of uncertainty induced by the platform supporting the software of a CPS, and in Section \ref{sec:dealing-with-uncertainty} we review the engineering methods used for analysis and implementation, aiming at reducing the risks due to uncertainty. These two sections show that reducing or eliminating uncertainty increases the complexity of a design and, ultimately, its cost. 
In Section \ref{sec:autonomy-new-uncertainty} we discuss the need for more autonomous systems, the new challenges that they bring in terms of analysis and implementation, and the need for new methods and tools that support compositional reasoning, as well as the modeling of aleatoric and (equally importantly) epistemic uncertainty. Finally, in Section \ref{sec:two-key-elements-of-new-modeling} we present a design methodology and a modeling paradigm that addresses these needs.

\section{The Evolution of Cyber-Physical Systems}
\label{sec:evolution-cps}
In the last decade, we have witnessed an evolution of the architecture and application domains of CPS. Systems that control physical processes are functionally decomposed into a plant and a controller. The plant comprises a set of physical variables that evolve over time according to certain dynamics. The controller observes some of these variables through sensors, tracks their evolution over time, and computes actions that are translated into physical effects by actuators. All these elements have evolved in complexity, from control systems to the most advanced form of autonomous systems such as self-driving vehicles and teams of Unmanned Aerial Vehicles (UAVs).

\emph{In traditional control systems}, the implementation platform is an embedded system with analog and digital Input/Output (I/O), a processing unit, local storage, and communication interfaces. The I/O are connected to sensors and actuators, and the estimation and control laws are implemented as a software function that runs periodically according to the time-scale of the physical process to be controlled. In these systems, the model of the plant is typically known. The controller is designed using traditional methods such as Lyapunov-based tools.

Control systems can be connected over a network to form a \emph{networked control system}. In the simplest deployment, sensor readings are sent to a controller in the form of packets over a network, and the computed control outputs are sent back to actuators also in the form of packets over a network. However, the control architecture could become more complex. For example, many applications such as cars, aircraft, or buildings,  span large physical plants featuring several interacting sub-systems. In these cases, control systems are typically federated: many local sensing and actuation points are connected to a controller, which in turn may be connected to a supervisory controller. The functionality of a sensing point is typically signal conditioning, or computation of threshold crossings, while actuation points implement a local controller for a physical device that maintains a given set point. The set point is computed by the supervisory controller according to some pre-defined policy. This kind of architecture can be found in building controls for instance, where the supervisory controller decides set-points for chillers and variable air volume terminal units depending on occupancy and time of day, while local controllers track the temperature of individual zones and control the mix of cold/hot air. The communication demand over the network that links these systems is low, communications can be considered reliable, and the delivery of messages from one system to another is not subject to stringent delay or jitter requirements.

As connectivity became more pervasive, and embedded platforms less power hungry, more capable, and less expensive, local controllers (even if elementary such as switches) could be networked as well. The \emph{internet of things} is a term that was coined at the edge of the current century. Being embedded in physical objects, these new class of CPS are characterized by addressable physical things at larger scale (millions, billions or even tens of trillions of devices \cite{dahlqvist2019growing}). Connectivity gives rise to new opportunities as the massive amount of data collected by these devices can be processed to create models using new machine learning techniques for prediction, maintenance, and optimization. However, connection among nodes can no longer be assumed reliable. In these types of systems, new nodes can join the network while others can leave, and nodes need to be able to describe their capabilities. 

Moving along the functional axis, CPS have evolved to automate more complex functions that have been traditionally under the responsibility of humans. \emph{Autonomous systems} are a natural evolution of CPS, and autonomy is a feature orthogonal to the evolution trajectory described above. Autonomous systems, in fact, can be localized or distributed, and networked over a range of communication technologies. We can contrast traditional controlled systems shown in Figure \ref{fig:traditional-controlled-systems}, and autonomous systems shown in Figure \ref{fig:autonomous-multi-agent-systems}. In control systems, the environment is bounded by an envelope defined as constraints on physical variables. It is also bounded in terms of the way in which humans can interact with the system. Human inputs are encoded as a finite set of well-defined commands, while outputs are values of state variables and modes of the systems displayed to the user. The human-machine interaction, in fact, is mediated by a user interface that can be implemented by a display and a set of switches or numerical set points. All higher level commands, which could represent a large class, are elaborated by humans that plan for lower-level commands sent to the system. The controlled system senses physical variables and updates the state estimate using accurate models. A control law, also derived using models, maps the belief state into actuation that can be applied with known error bounds. Control laws operate at short time scales which makes predictions fairly accurate. The state is typically encoded as a vector of continuous state variables representing the physical variables in the environment. They are also typically associated with a known probability density function. In some cases, a supervisory control policy may change the control law in a pre-defined way depending on the command received by the user or on the state of the CPS.    

\begin{figure}
    \centering
    \includegraphics[width=0.7\textwidth]{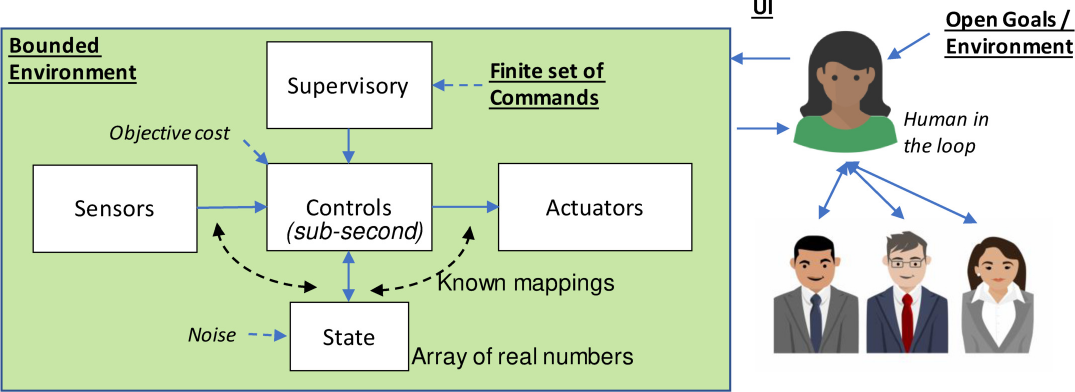}
    \caption{A traditional controlled systems}
    \label{fig:traditional-controlled-systems}
\end{figure}

\begin{figure}
    \centering
    \includegraphics[width=0.7\textwidth]{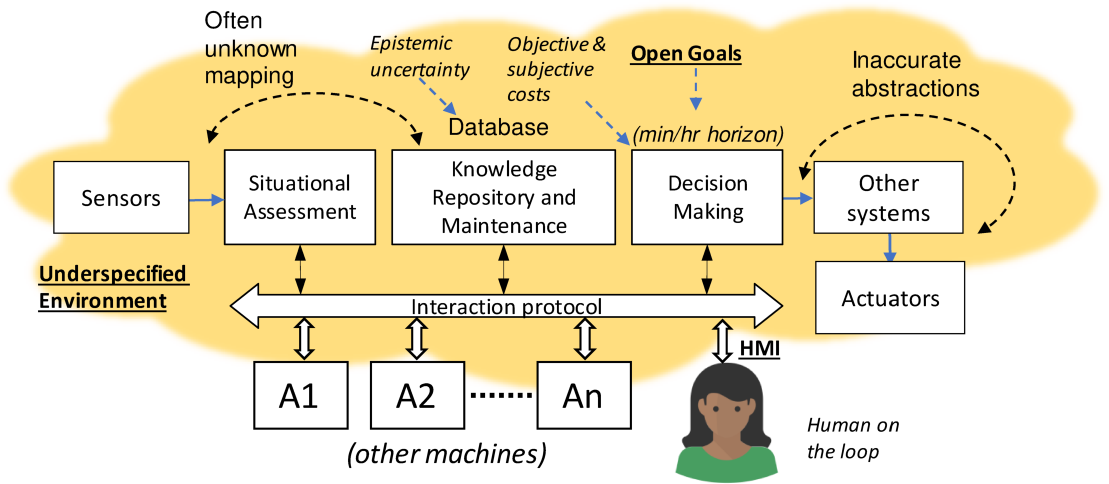}
    \caption{A multi-agent system featuring autonomous agents with local situational awareness and decision-making capabilities.}
    \label{fig:autonomous-multi-agent-systems}
\end{figure}

The shift of tasks from humans to the autonomous system exposes a set of other challenges. First, the environment and the goals that were assessed and understood by humans, have to be now understood by the system. Both the environment and goals become hard to bound at design time, resulting in an under-specified environment and an open set of goals. The autonomous system is in charge of using the sensed information to assess the situation. However, the semantic gap between the sensed variables and what they mean in terms of state of the environment is much higher than in control systems. A classic example is the identification and classification of objects from images. In many cases, models that related sensed variables to situations are not known, and are approximated by parametric or non-parametric models learned from data. Goals have to be understood and then decomposed by the system through planning. Plans span a much longer time horizon than control systems and predictions become much more uncertain. Moreover, optimality is judged using subjective measures: cost functions are defined by humans to encode preferences such as choices that are more or less ethical or safe, they are evaluated over an internal world model, and they are projected in the future with high uncertainty. Even the effect of actions is hard to model as it is the result of abstractions that are needed to reduce the complexity of the planning problem. Finally, autonomous systems are also typically networked and interact by negotiating actions and sharing information contributing to a potential explosion of uncertainty in the global state. 

\section{Sources and types of uncertainty}
\label{sec:uncertainty}

Previous work has investigated the sources and types of uncertainty to be considered in the design of CPS. 
For example, the work in \cite{hastings2004framework} provides a taxonomy of uncertainty, risks (or opportunities), mitigation strategies, and outcomes. 
Uncertainty can arise from lack of knowledge or lack of definition, and it can be either statistically characterized, a known unknown, or an unknown unknown. 
The risks associated with uncertainty can range from catastrophic to system performance degradation, they can have impact on cost and schedule, or they can present new opportunities such as emergent capabilities. 
Typical mitigation strategies include margins, redundancy, generality, and upgradability. 
The outcomes include reliability, robustness, flexibility, and evolvability. 
The work in \cite{ramirez2012taxonomy} provides a taxonomy of uncertainty introduced either at the requirements, design, or run-time stage of a system lifecycle. 
Requirements could be incomplete, ambiguous, or unsatisfiable. 
The design could be inadequate because of unexplored alternatives, or lack of traceability. 
At run-time, the system may experience failures, unpredictable environments, or inaccurate estimate of the state of the world.

For formal analysis, a mathematical framework is required for representing uncertainty. 
There are two types of uncertainty that should be considered: aleatoric and epistemic. 
Aleatoric uncertainty is inherent in the observed phenomenon and cannot be reduced. 
For example, the measurement of a physical quantity is the combination of the actual value plus an additive noise term that depends on several factors such as thermal noise. In this case, the measurement signal is represented by a stochastic process which can be characterized by its properties such as ergodicity and stationarity, and by its statistics such as mean and autocorrelation, or even by a probability density function. 
In other cases, uncertainty is related to the parameters of a system. For example, the mechanical properties of a system which directly contribute to the parameters of the model of a plant might not be exactly known due to the manufacturing process which cannot be accurately controlled. 
It is, however, far from uncommon to incur in lack of knowledge about uncertainty, and other models might be more appropriate such as the Dempster-Shafer theory of evidence \cite{dempster_upper_1967,shafer_mathematical_1976} or fuzzy sets \cite{zadeh_fuzzy_nodate}. Other methods that directly address epistemic uncertainty will be discussed later in Section \ref{sec:two-key-elements-of-new-modeling}.

While these representations are useful to develop algorithms for formal analysis of uncertain systems, their application to system design requires additional methods and tools that define relevant quantities, their associated uncertainties, and efficient modeling methodologies. In a model-driven development process, the functional view of the CPS is an architecture which includes the environment, the variables that are sensed, the block diagram of the estimation and control functions, and the actuation signals that are sent back to the environment. This model is then implemented by what we have named the \emph{platform} where sensed variables are mapped to sensors, the estimation and control functions are mapped to software modules, and the actuation signals are mapped to actuators. Links among functional blocks are implemented by communication paths that could be simply variables held in memory, or connections over a network. The semantics of the model determines the execution policy of the software that is supported by an operating system. The platform includes several sources of uncertainty that affect availability and integrity of data, correctness, and precision in actuation. We now list some sources of uncertainty introduced by a platform supporting a CPS application together with models used for analysis.

\noindent
\textbf{Failures}. All components are subject to failures. A component such as a sensor, an actuator, a processor or any other electronic or mechanical component can be characterized by the probability distribution $F(t)=P(T \leq t)$ of its time to failure $T$. This distribution defines the probability that a component fails before $t$ \cite{kapur2014reliability}. Let $f(t)$ denote the corresponding probability density function. A metric often used in reliability models is the Mean Time To Failure (MTTF)  (or Mean Time Between Failures if the component is serviceable) which is the first moment of $f(t)$. 

Components may fail in different ways, each having a different impact on system safety and performance. For example, a sensor may stop functioning, or may generate erroneous measurements. These two different failure modes have different consequences and require different identification mechanisms at run-time. The Failure Mode and Effect Analysis (FMEA) \cite{stamatis2003failure} methodology is used for this type of analysis. 

Fault Tree Analysis \cite{stamatelatos2002fault} is used to determine the probability of critical events to occur as a result of failures. In a fault tree, a top-level event of interest is decomposed recursively into intermediate events until basic events are reached that cannot be decomposed further. Starting from the failure statistics of the basic events (which must be available), quantitative analysis can then be used to compute the probability of the top-level event. 

\noindent
\textbf{Noise or imprecision}. Sensors measure physical quantities and generate an electrical signal. Such signal reflects the physical quantity up to a certain accuracy, and with added noise. Actuators are also imprecise: there is a difference between the electrical signal and the achieved physical effect. Computation may have to be implemented using finite precision hardware. Finally, the parameters of the plant model may not be precisely known, or the model used by the controller may bear errors due to abstraction.

\noindent
\textbf{Timing}. In a cyber-physical systems, the  execution of the control algorithm (that may include software, hardware and networks) must follow the same tempo of the physical environment. In other words, execution time is a fundamental contributor to the correctness specification of the control algorithm. Execution time from sensors to actuators is difficult to predict in general because it is affected by several factors including optimizations performed by compilers when generating machine code, delays induced by memory access which varies depending on the location of data and instructions (main memory, or cache), and instruction processing pipelines. Analysis tools aim at estimating the Worst Case Execution Time (WCET) of a task \cite{wilhelm_worst-case_2008}. For networked applications, the communication delay between two tasks depends on the type of protocols used by the network. Predicting communication delays is also often done using worst-case analysis as in Network Calculus \cite{le_boudec_network_2001}. Dedicated techniques can be used if the protocol is time-triggered, or priority-based.
For other protocols such as WiFi, probabilistic models based on Markov Chains are also available for analysis (see early work in \cite{bianchi_performance_2000} for example). Given the execution semantics of a set of communicating tasks, these models can be used to predict the end-to-end delay of a control function.

\noindent\textbf{Security}. For cyber-physical systems there are both cyber and physical security concerns. Example of cyber-attacks include standard attacks on communication networks. However, more complex attacks may come from the supply chain where counterfeit components containing malicious elements may end up being used in actual products. The physical side of the system can be attacked either by directly tampering with sensors, or by forcing the environment to go outside the envelope for which the system has been designed. Attacks can be characterized by different metrics including their likelihood (depending on how hard it would be for the attacker to succeed), cost of their consequences, and cost to prevent potential attacks. Attack trees \cite{mauw2005foundations}, which are similar to fault trees, can be used to analyze the cost and likelihood of an attack (which often involves multiple steps). 

\noindent\textbf{Implementation errors}. Hardware and software components may include implementation errors. We  consider the case where the requirements or specification of the hardware or software component is correct, but the implementation introduces errors. Such errors are not acceptable, and they are not statistically  characterized as in the case of failures. Rather, design processes and architectural solutions are used to minimize the risk of errors going undetected in the final implementation.

\section{Dealing with uncertainty: A short review}
\label{sec:dealing-with-uncertainty}
The sources of uncertainty coming from the platform described in the previous section impact the safety and performance of a CPS. However, engineering methods are able to deliver remarkably safe systems. For example, the probability of loss of control of a commercial aircraft is indeed $10^{-9}$ per flight hour as shown by the statistics kept updated by The Boeing Company \cite{airplanes2019statistical}. Design processes and tools have been developed to identify and remove uncertainty towards delivering high-assurance systems. An example of guidelines can be found in aviation where standards such as the ARP4661 \cite{arp47611996guidelines} and ARP4754 \cite{arp4754a2010guidelines} provide considerations on how to analyze, categorize and deal with hazards, and how to derive requirements that, when implemented, minimize the risk of catastrophic events.

Analysis and elimination of uncertainty is not achieved with the use of one system model, but rather with dedicated techniques for different aspects of the design, some of which have been referenced in the previous section. Each model focuses on the effective computation of key performance parameters, and on the qualitative or quantitative analysis of uncertainty. 

\noindent\textbf{Dealing with failures: guaranteeing availability and integrity.} Redundancy is a way to deal with failures. When the measurement of a physical quantity is essential to safe operations, two or three sensors can be used to deal with either permanent failures or the generation of wrong measurements. The outputs of these sensors are the inputs of a voting scheme that decides which measurement to trust. Redundancy is also used for computing platforms where multiple processors execute functionally equivalent software. To avoid common mode failures, processors, programming languages, and compilers may also be different, as demonstrated in the implementation of the Boeing 777 primary control system \cite{yeh2001safety}. Redundancy is a basic mechanism also used in communication systems where hosts are linked by multiple paths with independent links, and control bits are added to encoded information to detect and correct errors at the receiving end.

\noindent\textbf{Guaranteeing security.} Denying physical access to potential attackers is one possible measure that can be effective in some settings. However, as systems become more connected, and supply chains more global, the attack surface expands and other methods are needed such as the use of secure operating systems \cite{klein_sel4_2009}, and techniques to control the pedigree of hardware components \cite{rostami_primer_2014}. Physical attacks, require not only increasing the observability of the plant through redundant sensors, but also designing monitors to identify and isolate potential attacks \cite{pasqualetti_attack_2013}. 

\noindent\textbf{Guaranteeing timely execution.} The total delay from sensors to actuators must be bounded and small enough to follow the time-scale of the physical processes to be controlled. There could be uncertainty in the execution time of software and in the communication delay among hardware and software components. To remove this uncertainty, several techniques can be used both at the hardware and software level. In many cases, time determinacy is favored over performance and many useful advanced implementation methods such as compiler optimization, caching, and dynamic memory allocation are not used. These restrictions are needed to derive a good and reliable estimate of the worst-case computation time of tasks executed on a CPU. 

Once the computation time is known, schedulability analysis provides the framework for priorities to tasks.
Consider a set of $N$ tasks, each characterized by a period of execution $T_i$ (which in this simplified setting also corresponds to the task deadline) and worst case execution time $W_i$. A scheduling policy assigns priorities to these tasks so that they don't miss their deadlines. Typical policies include the rate-monotonic static priority assignment, and the earliest deadline first (EDF) dynamic priority assignment schemes. For both, there are simple schedulability tests that can assure that no task will miss its deadline. For example, when the EDF policy is used, a utilization $U=\sum_{i =1}^{N} W_i/T_i$ less than or equal to one guarantees schedulability. In practice, the total utilization is kept much lower than one (a case of margin) to deal with potential uncertainty. Tasks are then scheduled by a Real-Time Operating System (RTOS) which provides hard guarantees on typical services such as context switching (which is instead uncertain for general operating systems).

When communication among tasks occur over a network, the network delay must also be bounded and not uncertain. Communication protocols for cyber-physical systems have been developed to guarantee bounded transmission time. Some protocols allow for the definition of priorities associated with messages, and priority assignment can be used to bound the maximum delay experienced by a message. Other systems, such as the Time Triggered Protocol (TTP) guarantee time determinacy by assigning periodic time slots to the nodes in a network. 

\noindent\textbf{Being robust against noise and imprecision}. Filtering, and robust control \cite{zhou1998essentials} are common techniques to deal with bounded uncertainty about the model of the plant, and to reject several types of disturbances. 

\noindent\textbf{Developing bug free software.} Several strategies are put in place in a typical design flow to find and remove errors such as code reviews, unit testing, and integration testing (where testing is done in simulation, using Hardware-In-the-Loop, or directly in the field). For high-assurance systems, independence and redundancy are two commonly used strategies: the software development team is kept isolated from the test and evaluation team, the same function is written in different programming languages, compiled using different compilers, and executed at run-time by different processors. 

Several steps in the development process can also be automated. For example, Modified Condition / Decision Coverage \cite{hayhurst_practical_2001} is a metric that defines how well a set of tests triggers all conditions and decisions in the code. Tools for the generation of tests aiming at maximizing coverage also exist. Finally, formal method tools such as model-checking \cite{jhala_software_2009} and abstract interpretations \cite{bertrane_static_2015} can provide absolute assurance of software correctness. Interactive theorem provers, such as Isabelle/HOL\cite{Nipkow-Paulson-Wenzel:2002}, are also commonly used tool. Isabelle, for example, has been used to prove the functional correctness of the seL4 microkernel \cite{klein2009sel4}.

The effectiveness of a testing-based approach depends on the set of tests that are used to verify that a system satisfies requirements. In many cases, it is impossible or impractical to test a system for all possible inputs, which makes testing an incomplete verification method. However, testing can be done on the real system which does not require the development of additional models and increases confidence that the system will work in operation (at least for those conditions that have been used during testing). On the other hand, formal verification techniques are exhaustive and can produce a formal proof that the system satisfies its requirements. However, a formal model needs to be developed, and verification algorithms typically do not scale well with the size of the input and state encoding. 

\paragraph{Final Remarks}

Before closing this section, it is worth mentioning modeling, analysis, and synthesis tools that in principle can be used to formally verify properties of a cyber-physical system under uncertainty, or to guarantee the satisfaction of such properties by construction. 

Stochastic hybrid systems \cite{lavaei2021automated} is a modeling paradigm for systems that can be characterized by a finite set of discrete states, and where each state is associated with a noisy dynamical system (a stochastic differential equation). Probabilistic transitions among locations are associated with a guard that enables them, and a probabilistic reset map that changes the continuous state variables from the source state to a probability distribution in the target state. Typical analysis problems include reachability, safety, and stability, while design problem include optimal control. In principle, this model could accommodate uncertainties in the availability of sensors, and various delays. However, the failure of a sensor, for example, is a discrete event that requires introducing additional discrete states whose number increases exponentially with the number of failures, limiting the scalability of these methods.
    
Correct by construction methods can also be used to automatically explore a design space and find design points that automatically satisfy typical constraints such as maximum communication delay, and other physical constraints on the implementation platform. These methods encode the design problem as a set of constraints whose solution is an actual design. Optimization (or heuristic search) is used to find optimal solutions. The flexibility gained by formulating the design exploration problem as an optimization problem allows for incorporating margins, redundancy, and other preferences into the problem formulation. Examples can be found in our previous work \cite{pinto2007communication,leonardi2009case,leonardi2011synthesis} in the context of synthesis of computation and communication platforms for networked embedded systems. 

Additional measures can be designed to further reduce the impact of uncertainty on the operations of a cyber-physical system. Clearly, the environment can be designed in such a way to remove uncertainty. For example, warehouse robots such as the ones used by the Kiva system, rely on fiducial markers \cite{enright2011optimization}. The operations of a system can also be restricted so that uncertainty does not lead to safety violations. For example, a flight plan for a commercial flight must guarantee the presence of a suitable alternate airport within a certain time from any point in the plan under a single engine operation \cite{chiles2007etops}. 

\section{Autonomy as a driver of new classes of uncertainty}
\label{sec:autonomy-new-uncertainty}
There is strong interest in automating tasks that are today done by humans and that are considered dull, dangerous, where machines could perform better, or where humans represent a significant operating cost. Consider the case of commercial aviation. Two pilots are required for domestic flights, and three or more are required for longer flights. The need for pilots not only increases training and operating costs, but creates scheduling challenges as crews need rest time and cannot fly more than a certain number of hours per month. Moreover, a significant number of incidents are due to human error. For example, it is estimated that the cost of runway incidents is tens of billion dollars per year \cite{maurizio_anichini_solutions_2017}. 

Autonomy is not only a desire, but it is also a need in many application domains. 
For example, humans cannot match the pace of automation in assembly lines, in detecting and reacting to critical events, or in controlling many devices in parallel. 
In domains where applications cover a large physical space, individual systems find themselves far from other systems, and yet they must continue the execution of coordinated plans. 
While wireless communication technologies have improved, and bandwidth abound, there are still availability challenges that prevent the deployment of centralized architectures. 
This is for example the case of the envisioned MOSAIC warfare environment \cite{defense_advanced_research_projects_agency_darpa_darpa_nodate}, and the JADC2 framework \cite{jim_garamone_joint_nodate}. In these environments, communication (which is often contested) is a scarce resource. Similar problems exist even in the case of the Internet-of-Things due to the limited size and power of connected objects, and to the potential for attacks. 
Thus, nodes in these networks must be able to sense and understand their surrounding environment, make decisions that require fast response time locally, predict the global state of the system, and maintain the ability to coordinate within coalitions. \emph{Autonomy, then, is seen as one way to deal with disruption and long, unpredictable delays}. 

For these new class of systems, traditional tools to deal with uncertainty, such as the ones mentioned in Section \ref{sec:dealing-with-uncertainty} are no longer sufficient. They have been developed to deal with  uncertainty that is well-characterized such as noise in measurements, and equipment failures with known failure modes. Moreover, systems, even when distributed, are deployed in closed environments where communication can be implemented through reliable media and timing can be controlled. And finally, system integrators control the interfaces among sub-systems from different suppliers thereby guaranteeing some level of integrity of the information in and out of a system. In the case of autonomous systems, while these types of uncertainties are still present, epistemic uncertainty seems to be predominant. This type of uncertainty, is in principle reducible, albeit at a cost. 

The need for understanding the environment demands for complex sensors. An autonomous car, for instance, is equipped with LiDARs, cameras, radars, and ultrasonic sensors.  As already mentioned in Section \ref{sec:evolution-cps}, the mapping between the physical quantities measured by these sensors and logical facts useful for decision-making is not known in a form that can be analyzed and implemented. This means that, while locally a system still senses physical quantities such as the time of flight of a laser beam, it needs to be able to deduce high-level facts such as the intent of another actor in the environment. To overcome this problem, approximations of these functions are learned from data \cite{bengio2017deep}. The predictions made by these models are subject to both aleatoric and epistemic uncertainty \cite{hullermeier_aleatoric_2021} which can be reduced by increasing the complexity of the models and the data set used for training, but both have an impact on complexity, and cost and time to deployment. The NVIDIA Drive AGX Pegasus \cite{noauthor_nvidia_nodate}, a platform for self-driving cars, is in fact a super-computer, delivering hundreds of tera-operations per second, and consuming hundreds of watts of power. Additional reasoning is required to deduce facts, make predictions about the environment, and ultimately plan for actions. The  complexity of the reasoning and planning algorithms depends on the complexity of the representation and the number of actions that can be taken by the autonomous agent. Reducing such complexity requires abstraction which adds uncertainty as some information must be lost. Moreover, the world in which these systems operate is very dynamic, meaning that the rules according to which the environment behaves may change over time. This is due to other agents in the environment that learn and adapt continuously in response to the behavior exhibited by a system. This is different from the case of a cyber-physical system whose behavior is driven by physical laws that can be considered immutable. Finally, not only communications among agents is imperfect, but systems can be attacked in many ways (both on the cyber and physical side \cite{eykholt_robust_2018}), meaning that the information they receive through the network or even through sensing of the physical world may be corrupted by malicious agent. These sources of epistemic uncertainty can be reduced by using complex models, complex communication schemes, or even by deploying infrastructure in the environment to simplify perception and situational assessment, but these solutions have an impact on the complexity and cost of the computing and communication platforms.

The additional complexity of autonomous cyber-physical systems brings new challenges to systems engineers. Guaranteeing availability through physical redundancy requires replication of expensive sensors and hardware modules. Falling back on less performing sensors or hardware is a potential solution if the system can be shown safe under the resulting lack of information. The identification of corrupted information requires additional sensors and complex models. The ability to guarantee a timely reaction to deal with contingencies or changes in objectives is not always possible as the running time of some tasks, such as optimization algorithms, is hard to predict and the worst case execution time hard to obtain or too conservative. The software implementation of the functionality of an autonomous system is hard to assure. The programming languages and data structures used in these applications, the heavy reuse of large third-party library, and the mapping over heterogeneous multi-processor platforms, render verification a challenging task. As a consequence, not only bugs may go undetected in the final product, but the code may fail in subtle ways. 

The discussion points in this section suggest that autonomous CPS must be designed to deal with epistemic uncertainty arising from abstraction, approximation, and cost constraints. The design methods should explicitly capture lack of knowledge and its impact on the possible behaviors of a system. It is not necessary for an autonomous agent to have perfect knowledge about the environment, but rather enough knowledge to make the right decisions so that its behaviors achieve desired goals while satisfying given constraints. Given the inherent cost of removing epistemic uncertainty, we are interested in the requirement generation process which drives functional design and implementation.

\section{The key elements of a new modeling paradigm}
\label{sec:two-key-elements-of-new-modeling}
We first introduce a functional architecture that defines what needs to be modeled and helps to drive the process of requirement generation for the different elements of an autonomous cyber-physical system. We then propose a formal and compositional modeling framework based on multi-modal logic which we motivate based on few observations.

\noindent
\textbf{Testing alone seems to be impractical for autonomous systems.} A report from the RAND Corporation \cite{kalraDrivingSafetyHow2016} estimates that the number of miles to be driven by an autonomous vehicle to prove with confidence that the fatality rate is comparable to human-driven vehicles is 10 billion. The total number of miles driven by the Tesla autopilot from 2015 to 2020 is estimated to be 3.3 billion \footnote{\url{https://lexfridman.com/tesla-autopilot-miles-and-vehicles/}}. This statistics shows that it is impractical to rely on testing only. For this reason, we advocate the use of a formal framework that can complement testing in providing assurance guarantees.

\noindent
\textbf{Applications will span multiple autonomous systems operating under many types of uncertainty.} The specification language used in the modeling framework should be able to represent probabilistic and epistemic uncertainty in a system with multiple agents. Each agent holds its own understanding of the world which is defined formally by logical statements. For this reason, we rely on \emph{modal logic} where the truth of a statement about the world can be qualified by expressions such as ``agent $a$ believes that'', or ``according to agent $a$, the probability of $\phi$ is at least $\alpha$''. Other modalities can also be included in the language such as temporal modalities. In fact, we will introduce a multi-modal first order logic language for the specification of component interfaces.

\noindent
\textbf{Autonomous systems will result from composition and evolution of capabilities.} The development of complex autonomous systems requires integration of many components and capabilities such as sensors, situational assessment, decision-making, actuation, and communication. These components may be provided by different suppliers. Moreover, systems supporting complex operations such as Urban Air Mobility \cite{UrbanAirMobility2020} will feature several stakeholders such as vehicle developers, fleet operators, and service providers. Finally, autonomous capabilities will be deployed incrementally as new technologies become available and as operational data provide insights on the environment, needs, and limitations of systems. 

Thus, we seek a modeling framework that enables modular, scalable, and compositional analysis and design. This is a strong motivation to adopt a \emph{compositional modeling framework} where components are modeled using a formal specification of their interfaces. These interfaces specify what a component is supposed to deliver (the guarantees) in a give set of environments (the assumptions). Thus, our focus is on requirements as they drive implementation cost.
When components are specified at a level of abstraction that is detailed enough to be implemented, we discuss some promising \emph{uncertainty quantification} methods approaches that can help in verifying that the implementation is compliant with the specified contract.

\subsection{High-Level functional architecture}
\label{sec:high-level-architecture}

\begin{figure}
    \centering
    \includegraphics[width=0.7\textwidth]{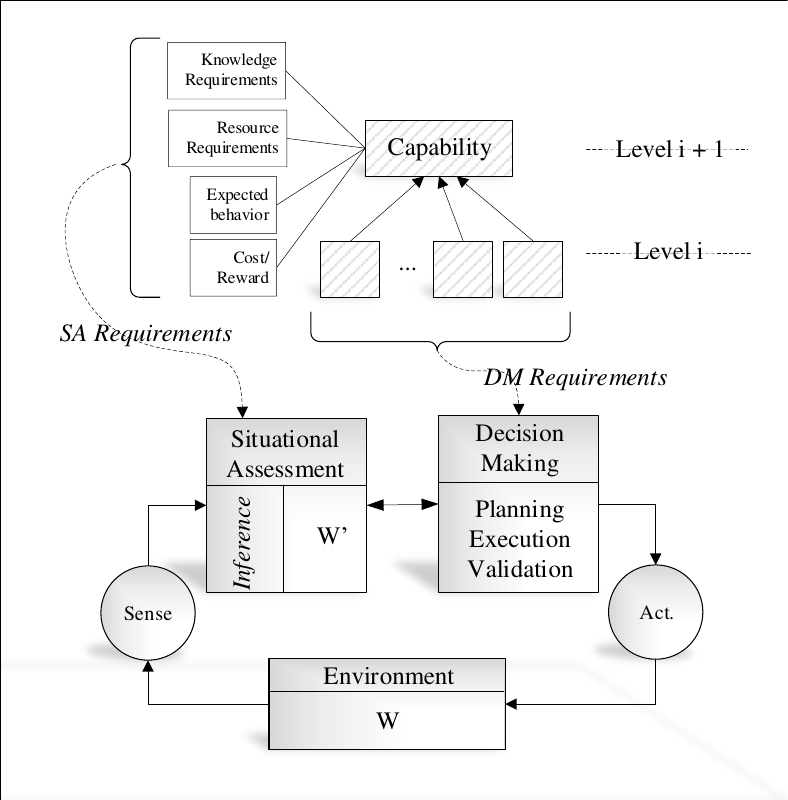}
    \caption{Overview of the high-level functional architecture of an autonomous systems (bottom), the capability model, and the requirement allocation process (top). \label{fig:method}}
\end{figure}

In this section we describe a methodology to drive requirement definition for autonomous systems. 
We do not address a comprehensive list of requirements such as operational requirements, maintainability, sustainability, security, cost and durability. We also don't address completeness\footnote{Completeness of requirements is in general difficult to define. Informally, a set of requirements ``is complete to the extent that all of its parts are present, and each part is fully developed'' \cite{boehm1984verifying}.}. 
Rather, the focus of this section is on requirements that are generated from an analysis of the functional architecture of a typical system as shown in Figure \ref{fig:method}. 
A system is a collection of agents, but the abstraction level at which we discuss requirements seats above such internal structure. 
Independently of the number of agents, an autonomous system has a natural functional decomposition which includes a \emph{situational assessment} function, and a \emph{decision-making} function (see the early work in \cite{endsley1995toward}, and our recent work in \cite{pinto2019open} among many others). 
The situational awareness function makes sense of inputs to create an understanding of the current state of the world, and to make predictions over a given time horizon. The decision-making function, instead, decides actions to take towards exhibiting certain behaviors and/or bringing about certain states.

At the fundamental level, an autonomous agent (Figure \ref{fig:method}) partially senses the environment state $w \in W$, and uses models to create an internal representation  $w' \in W'$. The situational assessment function establishes a relation between the two, that qualitatively can be expressed as $R_{SA} \subseteq W \times W'$ (and that will be discussed in more detail in Section \ref{sec:compositional-framework}). Each estimate $w' \in W'$ corresponds to potentially many actual situations $w \in W$. Because decision-making has only access to the internal representation of the world, defining requirements on $R_{SA}$ becomes essential. Notice that the estimate $w'$ does not need to be accurate, but only accurate enough to guarantee good decisions. The definition of such requirements, therefore, should start from an analysis of the decision-making process and the types of capabilities that the agents may decide to use.

To define capabilities and decision-making, we find inspiration in the capability approach from economics \cite{capframework}. A capability represents a potential that an agent has to realize a behavior or to achieve a certain state. A capability can be used only in certain contexts, namely when the agent has enough resources and is confident that the capability will succeed in its intent. For example, the capability of moving without collisions requires access to a map of the environment, and enough power to mechanically move the agent. Decision-making is the process of selecting the capability to use in a given state of the environment to achieve goals (that can be desired behaviors of states of the environment). Capabilities can be structured hierarchically where a group of capabilities at level $i$ can be leveraged together in such a way to deliver a capability at level $i+1$. For example, the capability of starting a car, moving both forward and backwards, avoiding obstacles, and following driving rules, can be used to drive from one location to another (although it should be proved that this is indeed the case).

At the lowest level of this hierarchy, we find capabilities that interface directly with the physical world. These capabilities are managed by control agents that provide a set of dependable automation functions which we will call \emph{skills}. Skills have a known set of failure modes that can be modeled and exposed to the next level of the decision hierarchy, and taken into account during the planning process. The time-scale at which skills operate is small enough that predictions can be made with a known degree of accuracy. Examples of skills include low-level controls such as the trajectory following capability of an autopilot.

At the higher-levels of the decision hierarchy we find strategic planning. At this level, the representation of the state space tends to be large and unstructured, time scales tend to become longer, and the actions of other actors must be taken into account. As a result, it is in general harder to provide accurate estimates or projections about the world.

Decision-making includes planning, execution, and run-time validation. In general, planning results in policies (i.e., when and where to use capabilities) that are executed by the agent, which also uses run-time validation to assess how plan execution is going, and forecast potential issues that may arise in the future. There are two sets of requirements that can be derived from the analysis of the capability model. First, the agent \emph{must know} if a capability can be used, and if and when it achieves its intended goals -- both essential for plan execution, monitoring and validation. These requirements directly apply to situational awareness. Secondly, the agent must be able to deliver a capability at level $i$ using capabilities at level $i+1$. This second set of requirements directly apply to decision-making which has the freedom to select policies using $i+1$-th level capabilities. Feasibility in all conditions as expressed by the $i$-th level capability, as well as optimality while respecting constraints must both be verified. Notice that these requirements are in general hard to verify because it is common to rely on the shape of the cost functions, and on the requirements for using a capability to enforce safety constraints. Verification, therefore, must check that under all conditions, and all applicable goals, the computed plan is valid.

\begin{example}[Obstacle avoidance scenario]
    In this example we capture informally the capabilities of the use-case described in the introductory chapter of this collection. The top-level capability $C^{DriveTo}$ requires knowledge of the current vehicle state $\hat{z}_k$, the goal location $g$, and a $map$ of the environment containing location and bounding box of obstacles. The expected behavior is to follow a trajectory that is collision free and that arrives at $g$. 
    
    The lower-level capabilities is a set comprising a trajectory optimizer using a kinematic model $C^{KM}$, a trajectory optimizer using a dynamic model $C^{DM}$, and a hybrid supervisor $C^{Sup}$. Both the kinematic model and the dynamic model require knowledge of $\hat{z}_k$, $g$, and $map$. The expected behavior of $C^{KM}$ is to finish its computation in $S_k^{KM}$ time steps, and to provide the sequence of optimal controls $U_k^{KM}$ according to the kinematic model $KM$. Similarly, The expected behavior of $C^{DM}$ is to finish its computation in $S_k^{DM}$ time steps, and to provide the sequence of optimal controls $U_k^{DM}$ according to the kinematic model $DM$. The hybrid supervisor requires knowledge of $\hat{z}_k$, $g$, $map$, $S_k^{KM}$, $S_k^{DM}$, and the actual trajectories $trj_k^{KM}$ and $trj_k^{DM}$ that the vehicle would follow when using $U_k^{KM}$ or $U_k^{DM}$, respectively. The expected behavior of $C^{Sup}$ is to select the control to be applied in order to avoid collisions and to reach the goal. 
    
    From this set of capabilities, we can derive situational awareness requirements. The internal world model is a tuple consisting of the vehicle state (position, velocity, heading, and steering angle), the map of the environment, and the trajectory that the vehicle follows under a given control sequence. The external world model has the same variables, and the situational awareness module must guarantee that the two are close enough. Such module will need to be refined as including perhaps a localization sub-system, a camera or Lidar sensor together with a mapping function that maintains an accurate map of obstacles, and a prediction module accurate enough to be used as reference model. 
    
    Given these capability specifications, components can be implemented to satisfy them. For example, the work in  \cite{125} proposes implementations for $C^{KM}$, $C^{DM}$ and $C^{Sup}$ where the first two use model-predictive control, and the last is a hybrid supervisor with a decision boundary between the two optimizers which depends on the velocity and steering angle of the vehicle. 
    
    There are several challenges that this informal example exposes. First, a formal language should be used to avoid ambiguities in the meaning of terms such as ``collision free'', and ``knowing something''. We also did not include robustness against uncertainty for example in localization or in the construction of the map. Assuming that such language exists, a proof should be provided that 
    $C^{KM}$, $C^{DM}$ and $C^{Sup}$ together satisfy the top-level capability $C^{DriveTo}$. Furthermore, a proof should be provided that the implementation of each component (such as the ones proposed in \cite{125}) is correct. These proof obligations should also be formally defined as verification problems. Section \ref{sec:compositional-framework} addresses the formalization of the specification language and the definition of the verification problems. 
    \end{example}

In addition to situational assessment and decision-making requirements, several other common requirements may need to be derived and assigned to agents, such as the ability to understand goals, to explain the reasons for failing to find a plan (or to achieve a goal), and to maintain safety in any state of the world. Furthermore, validation and verification of a plan at run-time should not be based on the same models that are used in planning which could otherwise lead to common-mode failures. Independence requirements between planning and run-time verification should be considered. In the case of learning agents, several requirements should be included such as the ability to identify good events to learn from, to compute the extent of the revisions of the internal knowledge of the agent, to perform the update, and to safely resume operations. The overall requirement is that the system improves, meaning that its performance improves at least in some environments, while safety guarantees remain unaffected.

Both situational awareness and decision-making functions result in general from the composition  of several sub-systems. For example, consider an autonomous vehicle. The situational assessment function will comprise components for dealing with road signs, pedestrians, and vehicle health. Each component will work under some assumptions on the environment such as good visibility. The decision-making function will comprise components for controlling comfort inside the vehicle, the entertainment system, and the trajectory of the car. These components will only work under assumptions on the health of the mechanical equipment, and the condition of the road. Initial versions of an autonomous system may only include a limited set of capabilities, to be expanded by adding additional functions in later versions. Furthermore, several teams may contribute to the design of a system, with teams perhaps even belonging to different companies. Incremental deployment and multiple suppliers have to be supported by a corresponding incremental analysis and design methodology.

In the next section we introduce a formal framework to support the modeling, analysis, and design approach described in this section.

\subsection{Compositional modeling framework and specification language}
\label{sec:compositional-framework}
We will use a contract-based approach \cite{benveniste_contracts_2018} to develop our compositional framework, and we will adopt the formalism in  \cite{bauer_moving_2012} which we briefly summarize here. A specification theory is a triple $(\mathcal{S},\otimes,\leq)$, where $\mathcal{S}$ is a set of specifications, $\otimes : \mathcal{S} \times \mathcal{S} \rightarrow \mathcal{S}$ is a parallel composition operator over specifications, and $\leq \subseteq \mathcal{S} \times \mathcal{S}$ is a reflexive and transitive refinement relation.  
The refinement relation must be compositional, meaning that given specifications $S$, $S'$, $T$ and $T'$, whenever $S\leq S'$ and $T \leq T'$, then $S \otimes T \leq S' \otimes T'$. A conjunction operator $\wedge : \mathcal{S} \times \mathcal{S} \rightarrow \mathcal{S}$ is used to combine specifications. The conjunction of two specifications $S$ and $T$, written $S\wedge T$, if it exists, is the most general specification that realizes both.
In some cases, a specification theory can be equipped with a quotient operator $/:\mathcal{S} \times \mathcal{S} \rightarrow \mathcal{S}$. Given two specifications $S$ and $T$, $T/S$, if it exists, is the most general specification such that $S\otimes (T/S) \leq T$.
Finally, relativized refinement is a ternary relation in $\mathcal{S} \times \mathcal{S} \times \mathcal{S}$ on specifications. 
Given three specifications $S$, $E$, and $T$, $S \leq_E T$ if and only if for all $E' \leq E$, $S\otimes E' \leq T \otimes E'$. 
In words, $S$ refines $T$ when both are considered in an environment $E'$ that refines $E$. 
It can be shown that relativized refinement is a preorder, meaning that for all specifications $S$, $E$, $E'$, and $T$, $E\leq E' \ and \ S\leq T \Rightarrow S \leq_{E'} T$.

A contract theory is built over a specification theory.
Specifically, a contract is a pair $C=(A,G)$ where $A$ and $G$ are specifications called assumption and guarantee, respectively. 
The environment semantics of a contract is the set of specifications that refines the assumption:
$\meaningof{C}_{env} = \{ E \in \mathcal{S} | E \leq A \}$. 
The implementation semantics of a contract is the set of specifications that satisfy the guarantee $G$ under assumption $A$: $\meaningof{C}_{impl} = \{ I \in \mathcal{S} | I \leq_A G \}$. 
Two contracts are semantically equivalent if their environment and implementation semantics are the same, respectively.
Moreover, a contract $C$ is in normal form if for all specification $I \in \mathcal{S}$, $I \leq_A G$ if and only if $I \leq G$, which also means that for a contract $C^{nf}=(A^{nf},G^{nf})$ in normal form, $\meaningof{C^{nf}}_{impl} = \{ I \in \mathcal{S} | I \leq G^{nf} \}$. 
In some cases, depending on the specification theory,  a contract $C=(A,G)$ can be transformed in a semantically equivalent contract $C^{nf}=(A,G^{nf})$ in normal form by weakening its guarantee\footnote{Under some assumptions on the specification theory, the normal form of the guarantee of a contract can be computed as $G^{nf}=G\wedge \neg A$ (where $\neg A$ is the set of environments that do not refine $A$). In this case, a contract in normal form is also called saturated.}.

Contracts are related by a refinement relation $\preceq$. A contract $C'=(A',G')$ refines $C=(A,G)$ if and only if $\meaningof{C'}_{env} \supseteq \meaningof{C}_{env}$, and $\meaningof{C'}_{impl} \subseteq \meaningof{C}_{impl}$. 
It can be shown that this condition corresponds to $A \leq A'$ and $G' \leq_A G$, and if contracts are in normal form, then the latter can be written as $G' \leq G$. 

Contracts can be composed to yield a new contract. Let $C_1 = (A_1, G_1^{nf})$ and $C_2 = (A_2,G_2^{nf})$ be two contracts, then their composition is $C_1 \boxtimes C_2 = ( A_1/G_2^{nf} \wedge A_2/G_1^{nf}, G_1^{nf} \otimes G_2^{nf})$. Finally, the quotient between $C_1$ and $C_2$, denoted $C_1 / C_2$ is defined as follows: $C' \preceq C_1/C_2 \Leftrightarrow C' \boxtimes C_2 \preceq C_1$. The quotient can be used to synthesize the specification of a missing component\cite{romeo2018quotient}.

It is sometime easier to prove whether a set of contracts $\{C_1,\ldots,C_n\}$ is a correct decomposition of a contract $C$. This can be shown \cite{le_contract-based_2016} to be equivalent to proving that:

\begin{align}
\bigwedge_{1\leq i \leq n} G_i^{nf} \leq G^{nf} & & \label{eq:gnf-ref}\\
A \wedge \bigwedge_{1 \leq j\neq i \leq n} G_j^{nf} \leq A_i,& &\forall i \in [1,n] \label{eq:ai-ref}
\end{align}

We now need to define our specification theory. 
We use a multi-modal first-order language that allows reasoning about knowledge and probabilities and that takes inspiration from \cite{fagin1994reasoning}. 
The use of a first order language allows us to also reason about objects, their properties, and their relationships which is convenient for autonomous systems. 
Let $\Sigma=(C,P,V)$ be a first order signature with the usual set of constant, predicate and variable symbols. 
We assume that the reader is familiar with the syntax and semantics of a first-order language \cite{fitting_first-order_2012,fitting_first-order_2012-1}. Also, let $\mathcal A$ be a set of agent symbols. The syntax of the specification language that we will use is recursively defined as follows:
\[
\phi := P^k(t_1,\ldots,t_k) | \neg \psi | \psi \wedge \psi'  | \forall v. \psi | \mathsf{K}_a \psi | \sum_{i=1}^{l} q_i \mathsf{P}_{a_i} \psi_i \leq b
\]
where a term $t_i$ is a variable or a constant, $P^k \in P$ is a predicate symbol of arity $k$, $\psi$ and $\psi'$ are formulas, $v \in V$ is a variable, 
$q_i$ and $b$ are rational numbers, and both $a$ and $a_i$ are agent symbols in $\mathcal A$. A formula $\mathsf{K}_a \psi$ denotes that agent $a$ knows that $\psi$, while  $\mathsf{P}_{a} \psi$ denotes the probability that agent $a$ assigns to $\psi$. 

The semantics of a formula in this language is the set of Kripke structures that satisfy the formula. A Kripke structure for a multi-agent system is a tuple $M = (W,\mathcal A,\{\sim_a\}_{a \in \mathcal A},\{P_a\}_{a \in \mathcal A},D,I)$ where $W$ is a set of possible worlds, $\sim_a \subseteq W \times W$ is the accessibility relation for agent $a$, $P_a$ associates to each possible world $w \in W$ a probability space $(W_{a,w},\mathcal F_{a,w}, \mu_{a,w})$,  $D$ is a fixed domain of objects, and $I$ is a function that maps each possible world $w \in W$ to a first order interpretation $I^w$ of the symbols in $\Sigma$. The satisfaction relation $\models$ that relates a Kripke structure $M$ and a world $w \in W$ to the formulas that are satisfied by the model is defined  as follows:

\begin{align}
(M,w) & \models P^k(t_1,...,t_k) & &iff & &(I^w(t_1),\ldots I^w(t_k)) \in I^w(P^k) \\
(M,w) &\models \neg \psi & &iff  & &(M,w) \not\models \psi \\
(M,w) &\models \psi \wedge \psi' & &iff & &(M,w) \models \psi \ and \ (M,w) \models \psi' \\
(M,w) &\models \forall v.\psi & &iff & &\forall o \in D, \ (M,w) \models \psi(v\leftarrow o) \\
(M,w) &\models \mathsf{K}_a \psi & &iff & &\forall w'\in W, \ w \sim_a w' \Rightarrow (M,w') \models \psi \\
(M,w) &\models \sum_{i=1}^{l} q_i \mathsf{P}_{a_i} \psi_i \leq b & &iff & & \sum_{i=1}^{l} q_i \mu_{a,w}( W_{a,w}(\psi_i) ) \leq b
\end{align}

Where $W_{a,w}(\psi) = \{w \in W_{a,w} | (M,w) \models \psi\}$, namely the set of worlds in $W_{a,w}$ that satisfy $\psi$. As discussed in \cite{fagin1994reasoning}, we can assume $W_{a,w}$ to be a subset of the set of possible worlds that the agent considers possible in $w$ (consistency). We will also assume that the probability spaces are the same for all agents, i.e. $P_{a_1} = P_{a_2}$ for all $a_1,a_2 \in \mathcal A$ (objectivity), and that for all $a$, $w$ and formula $\psi$, $W_{a,w}(\psi) \in \mathcal F_{a,w}$, meaning that all formulas define measurable sets. 

We define our specification theory as follows. A specification is a pair $S = (\Sigma, \phi)$ where $\Sigma$ is a first order signature, and $\phi$ is a formula in the specification language defined above over the signature $\Sigma$. The set of such specifications is denoted $\mathcal S_{KP}$. Given two specifications $S_1=(\Sigma_1, \phi_1)$ and $S_2=(\Sigma_2, \phi_2)$, their composition is $S_1 \otimes_{KP} S_2 = (\Sigma_1 \cup \Sigma_2, \phi_1 \wedge \phi_2)$. Moreover, $S_1 \leq_{KP} S_2$ if an only if $\Sigma_2 \subseteq \Sigma_1$, and $\phi_1 \models \phi_2$. 

\begin{proposition}
$(\mathcal S_{KP},\otimes_{KP},\leq_{KP})$ is a specification theory.    
\end{proposition}

\begin{proof}
It is easy to show that the refinement operator is reflexive and transitive. It remains to show that it is compositional. Let $S=(\Sigma_S,\phi_S)$, $S'=(\Sigma_S',\phi_S')$, $T=(\Sigma_T,\phi_T)$, and $T'=(\Sigma_T',\phi_T')$ be specifications such that $S \leq_{KP} S'$ and $T \leq_{KP} T'$. Then, $\Sigma_S' \subseteq \Sigma_S$, $\phi_S \models \phi_S'$, $\Sigma_T' \subseteq \Sigma_T$, and $\phi_T \models \phi_T'$ hold (premise). Also, we have that $S \otimes_{KP} T = (\Sigma_S \cup \Sigma_T,\phi_S \wedge \phi_T)$, and $S '\otimes_{KP} T' = (\Sigma_S' \cup \Sigma_T',\phi_S' \wedge \phi_T')$. From the premise, we can deduce that $\Sigma_S' \cup \Sigma_T' \subseteq \Sigma_S \cup \Sigma_T$, and that $\phi_S \wedge \phi_T \models \phi_S' \wedge \phi_T'$, which implies that $S \otimes_{KP} T \leq_{KP} S' \otimes_{KP} T'$.
\end{proof}

In the rest of this section, we will consider the same universal signature for all specifications. From this specification theory, we can build a contract theory as already described in the introductory notes to this section.

The last key element that we need to explore is the modeling methodology for autonomous multi-agent systems. The first order universal signature $\Sigma$ is partitioned into sub-signatures $\Sigma_{a_1},\ldots,\Sigma_{a_n}$ for each agent, such that $(C_a,P_a,V_a)$ is the signature for agent $a$. This signature contains the symbols that are used by the internal world model of the agent. The accessibility relation $\sim_a$ of an agent $a$ is defined as follows: $(w,w') \in \sim_a$ if and only if for all $s \in C_a \cup P_a \cup V_a$, $I^w(s) = I^{w'}(s)$. This means that in the two possible worlds $w$ and $w'$, the internal world model of the agent is the same, and therefore the agent cannot distinguish them. It is straightforward to prove that this accessibility relation  is reflexive, symmetric, and transitive (characteristic properties of the so called $S5$ system \cite{fagin_reasoning_2003}).

Equipped with this compositional framework that supports reasoning about knowledge and probability, and with a methodology to model autonomous systems, we can now present an example to show how it can be applied to the requirement analysis and generation phase of a design process. 

\begin{example}
Consider an unmanned aerial platform $u$ for urban air mobility arriving in proximity of a vertiport $p$ and having to perform a safe landing as shown in Figure \ref{fig:example_a}. The airspace is under the supervision and control of a local operator $c$. The environment $e$ includes different weather conditions. The vehicle is equipped with an on-board localization system that relies on the Global Navigation Satellite System (GNSS). In addition, a ground system also estimates the position of the vehicle which is communicated to $u$ through a communication link (LNK). 

\begin{figure}
    \centering
    \includegraphics[width=0.7\textwidth]{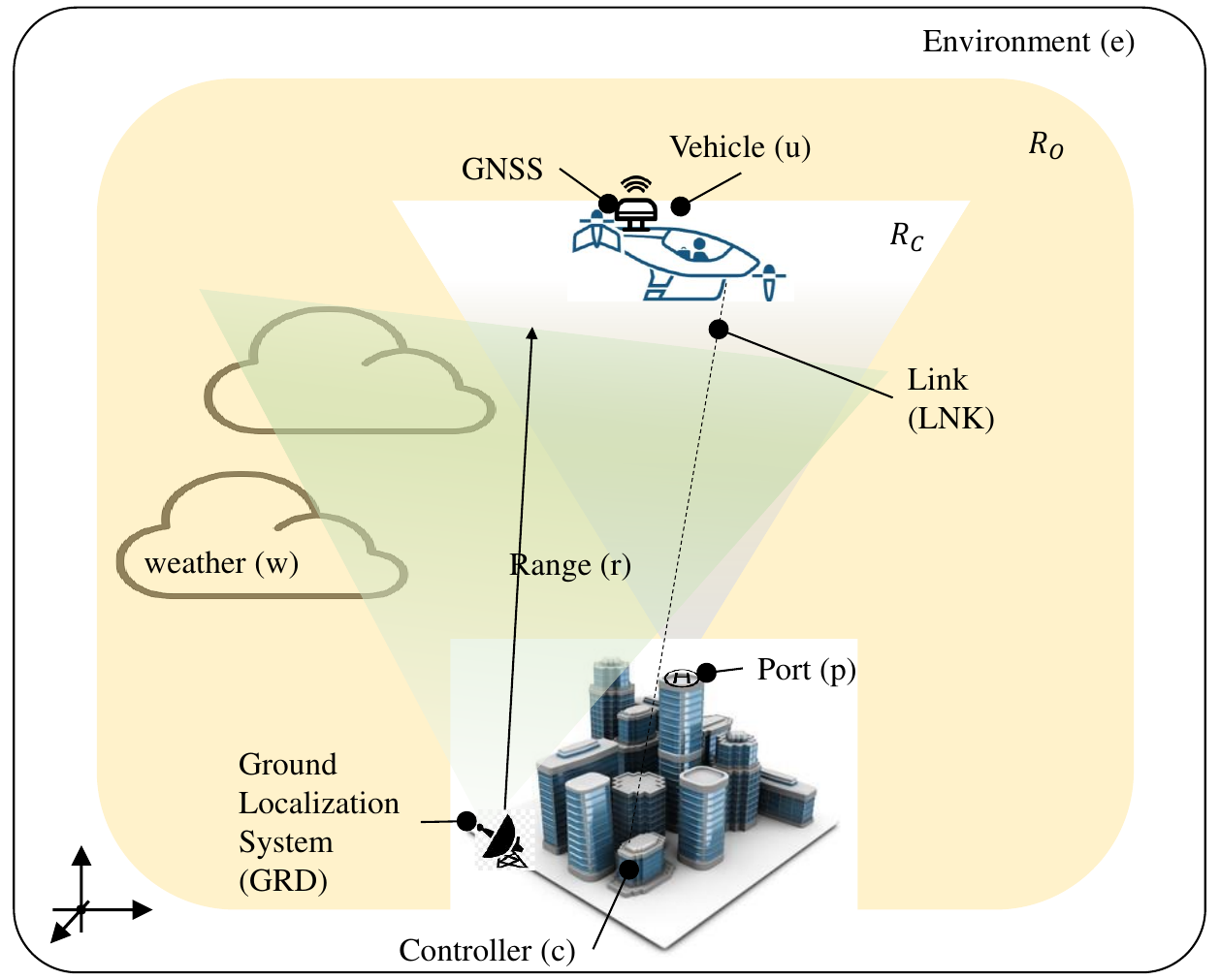}
    \caption{A scenario representing a vehicle tasked with landing at a local vertiport. The landing task is supported by situational awareness capabilities including a GNSS system on board, a ground support system (GRD), and a communication link between the vehicle and the local controller. \label{fig:example_a}}
\end{figure}

\begin{figure}
    \centering
    \includegraphics[width=0.8\textwidth]{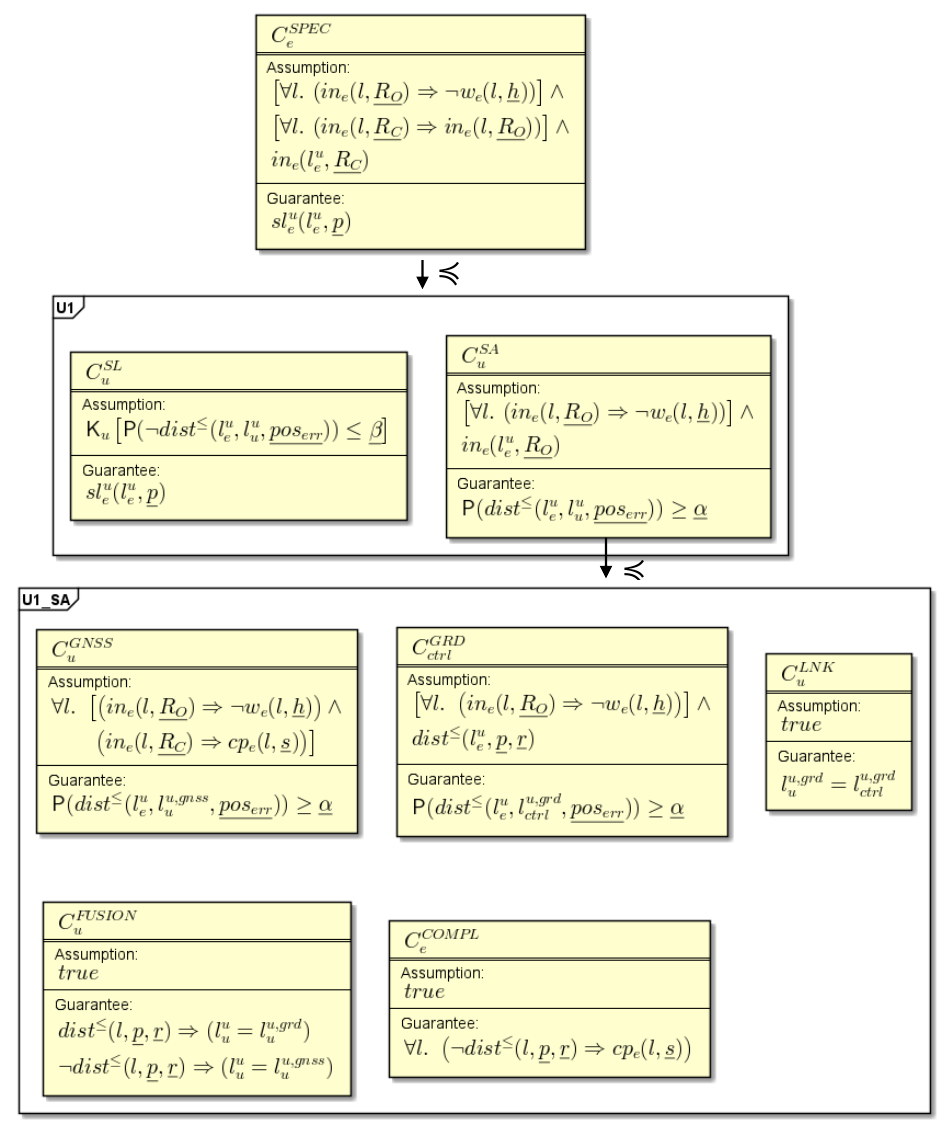}
    \caption{A compositional model of the scenario in Figure \ref{fig:example_a}. The top-level specification for safe landing in non-harsh weather conditions is refined into a situational awareness contract and a safe landing contract for the vehicle. The situational awareness contract is further refinement into the composition of contracts for the environment, the ground support, the on-board localization system, and the communication link. \label{fig:example_b}}
\end{figure}

In this example, we have a set of agents $\mathcal A$ = $\{u,e,c\}$ that denote a UAV, the environment, and a local operator, respectively. Symbols with no subscript have the same interpretation across agents,  and underlined symbols represent constants. 
Let us focus on a single top-level requirement modeled by contract $C_e^{SPEC}$. In order to model such requirement, we introduce a few symbols to reason about weather conditions, and to capture the goal of a safe landing. Let $w_e$ be a binary predicate such that $w_e(l,x)$ is true if the weather condition at location $l$ is $x$. We also introduce a few constants as follows: $\underline{u}$ denotes the vehicle, $\underline{p}$ denotes the vertiport, and $\underline{s}$, $\underline{m}$, and $\underline{h}$ are three constants used as a discrete representation of environment difficulty (simple, moderate, or hard, respectively). These levels will be used both for weather conditions and environment complexity (such as level of clutter or presence of obstacles). Complexity is modeled by a binary predicate $cp_e$, such that $cp_e(l,x)$ is true if the environment complexity at location $l$ is $x$. In addition, we assume that both $w_e$, and $cp_e$ are functional predicates, meaning that for each location $l$, there exists one and only one level $x$ ($x'$) such that $w_e(l,x)$ ($cp_e(l,x')$) is true. The approach region is denoted by $\underline{R_O}$, and the final approach region is denoted by $\underline{R_C}$. Predicate $in_e(l,r)$ is true if location $l$ is in region $r$. Finally, we use variables to model the location of the vehicle. For example, we will denote by $l_e^u$ the true location in the environment, and by $l_u^u$ the location estimate according to $u$. 

We can now state contract $C_e^{SPEC}$ as shown in Figure \ref{fig:example_b}. The assumption of this contract states several things: (1) that the weather conditions are not hard to handle in the approach region, (2) that the final approach region is inside the approach region, and (3) that the vehicle is in final approach. The system must guarantee a safe landing at the vertiport starting from its current location. Notice that the symbols used in this constraint belong to the environment because the requirement it represents must be satisfied in the real world.

According to the methodology described in Section \ref{sec:two-key-elements-of-new-modeling},  we decompose the specification of the vehicle into two contracts: a safe landing capability $C_{u}^{SL}$, and a situational awareness contract $C_{u}^{SA}$. The safe landing capability captures the ability to generate plans, to execute them, and to change them as needed during execution. To limit the complexity of this example, we only discuss assurance on the localization of the vehicle. To execute a flight plan correctly, the safe landing capability needs to be confident that its internal location estimate is good enough under all possible conditions. The assumption can then be modeled as\footnote{We have introduced a ternary predicate $dist^{\leq}$ such that $dist^{\leq}(l,l',d)$ is true if the distance between locations $l$ and $l'$ is less than or equal to $d$ (we are omitting the full axiomatization of such predicate as it is not needed for the purpose of this example).} $\mathsf K_{u} \mathsf{P}( \neg dist^{\leq}(l_e^u,l_u^u,\underline{pos_{err}}) ) \leq \underline \beta$: the vehicle knows that the probability that the distance between the true location and the estimate is greater than an acceptable position error $\underline{pos_{err}}$ is less than a constant $\underline{\beta}$. Under this assumption, the capability guarantees a safe landing at the vertiport. 

The situational awareness contracts $C_{u}^{SA}$ captures a requirement to maintain good localization with high probability. The assumption is that the weather condition is not hard in the approach region, and that the vehicle is in somewhere in such region. The guarantee is that the probability that the vehicle's position estimate $l_u^u$ is within $\underline{pos_{err}}$ from the actual position $l_e^u$, is greater than a constant $\underline{\alpha}$. As described in Section \ref{sec:two-key-elements-of-new-modeling}, this contract captures a constraint on the relation between the external (true) world and the internal estimate generated by agent $u$. 

\begin{proposition}
If  $1 - \underline{\alpha} \leq \underline{\beta}$, the set  $\{C_{u}^{SL},C_{u}^{SA}\}$ refines $C_e^{SPEC}$. \label{prop:refinement-sa-dm}
\end{proposition}
\begin{proof}
It is sufficient to prove the validity of the following formulas:
\begin{align*}
&A_e^{SPEC} \wedge (A_{u}^{SL} \Rightarrow G_{u}^{SL}) \wedge (A_{u}^{SA} \Rightarrow G_{u}^{SA}) \Rightarrow G_e^{SPEC} \\
& A_e^{SPEC} \wedge (A_{u}^{SL} \Rightarrow G_{u}^{SL}) \Rightarrow A_{u}^{SA}  \\
& A_e^{SPEC} \wedge (A_{u}^{SA} \Rightarrow G_{u}^{SA}) \Rightarrow A_{u}^{SL}  
\end{align*}
The second condition is trivially true becasue $A_e^{SPEC} \Leftrightarrow A_{u}^{SA}$. Becasue $G_{u}^{SL} \Leftrightarrow G_e^{SPEC}$, to show the validity of the first formula it is sufficient to show that $G_{u}^{SA} \Rightarrow A_{u}^{SL}$. We first show the following:
\begin{align*}
&\mathsf{P}(dist^{\leq}(l_e^u,l_u^u,\underline{pos_{err}})) \geq \underline \alpha \Rightarrow \\
& \mathsf{P}(\neg dist_{leq}(l_e^u,l_u^u,\underline{pos_{err}})) \leq 1-\alpha
\end{align*}
This is valid because $\mathsf{P}(\psi) + \mathsf{P}(\neg \psi) = 1$. From $1 - \underline{\alpha} \leq \underline{\beta}$, and the knowledge generalization axiom $\psi \Rightarrow \mathsf{K}_i \psi$ \cite{fagin1994reasoning} we deduce that $\mathsf{K}_{u}( \mathsf{P}(\neg dist^{\leq}(l_e^u,l_u^u,\underline{pos_{err}})) \leq \beta ) $.
\end{proof}

\begin{remark}
It may seem that the knowledge modality is not really important in the assumption of $C_{u}^{SL}$. Notice, however, that the situational awareness contract may induce potentially several possible worlds with different probability spaces relating the external and internal world model. For example, consider a contract $C_{u}^{SAf}$ that allows for potential failures that degrade the performance of situational awareness. The guarantee of this new contract is $(failure_e \Rightarrow \mathsf{P}(dist^{\leq}(l_e^u,l_u^u,\underline{pos_{err}})) \geq \underline \gamma) \wedge (\neg failure_e \Rightarrow \mathsf{P}(dist^{\leq}(l_e^u,l_u^u,\underline{pos_{err}})) \geq \underline \gamma')$. There are now two possible worlds, one where $failure_e$ is true and one where $failure_e$ is false, where the probability associated with a localization error above $\underline{pos_{err}}$ is different. Depending on the value of $\gamma$ and $\gamma'$, the safe landing capability may or may not know whether localization is accurate enough.
\end{remark}

The next refinement step shows the power of a compositional method. We are going to refine contract $C_{u}^{SA}$ and check that the refinement is valid, but we would not need to check that $C_e^{SPEC}$ is still satisfied. Moreover, we will show how requirements are generated and allocated to different systems in this scenario. To refine $C_{u}^{SA}$, we equip $u$ with a GNSS system modeled by contract $C_{u}^{GNSS}$. This system only works under simple or moderate weather conditions. It guarantees to provide a position estimate $l_u^{u,gnss}$ close to the true position with probability greater than $\underline{\alpha}$, but only in simple environments.  Notice that this contract alone, obviously does not refine $C_{u}^{SA}$. The GNSS system is complemented by a ground system under the control of the local operator and modeled by contract $C_{ctrl}^{GRD}$. This system works well even in cluttered environments, but has a limited range equal to $\underline{r}$. It provides a sufficiently accurate position estimate $l_{ctrl}^{u,grd}$. To simplify the example, we assume a perfect communication link, and also a perfect fusion algorithm. The communication link contract $C_{u}^{LNK}$ guarantees that the interpretation of $l_{ctrl}^{u,grd}$ (i.e., the location estimate on ground) is the same as the interpretation of $l_u^{u,grd}$ (i.e., the received estimate through the communication link) under all environments. The fusion contract $C_{u}^{FUSION}$ guarantees that the on-board system selects the source of the location estimate depending on the location of the vehicle. 

Following the same approach used in the proof of Property \ref{prop:refinement-sa-dm}, it is easy to show that the contract set $\{C_{u}^{GNSS},C_{ctrl}^{GRD},C_{u}^{LNK},C_{u}^{FUSION}\}$ does not refine $C_{u}^{SA}$. In particular, consider a location $l$ such that $\neg dist^{\leq}(l,\underline{p},\underline{r})$, and $cp_e(l,\underline{h})$. In this case, the assumptions of $C_{u}^{GNSS}$ and $C_{ctrl}^{GRD}$ are not satisfied, and therefore their guarantees don't need to hold. Thus, there would be no way for the guarantee of the contract set to entail the guarantee of $C_{u}^{SA}$. To address this problem it is sufficient to add a requirement on the environment modeled by contract $C_{e}^{COMPL}$. This contract guarantees that the environment at a distance greater than the range $\underline{r}$ from the port  $\underline{p}$ is simple. It can now be shown that the set $\{C_{u}^{GNSS},C_{ctrl}^{GRD},C_{u}^{LNK},C_{u}^{FUSION}, C_{e}^{COMPL}\}$ indeed refines $C_{u}^{SA}$. To check that this is the case, it is sufficient to see that each component sees its assumption satisfied by the assumptions of $C_{u}^{SA}$. Furthermore, the guarantees of $C_{u}^{SA}$ are either satisfied by the GNSS system at distances beyond $\underline{r}$ (given that the environment guarantees $cp_e(l,\underline{s})$ there), or by the ground system otherwise.
\end{example}

\paragraph{Remarks on analysis tools}
To make the methodology presented in this section effective, several tools must also be integrated in the design process. Two key analysis tools are needed: refinement checking, and component verification. Refinement checking consists in verifying that a set of contracts $\{C_1,\ldots,C_n\}$ refines a specification contract $C$. As mentioned in this section, refinement checking can be reduced to validity checking, or satisfiability checking, in the specification language for knowledge and probability. Several axiomatizations are available for similar languages (see for example \cite{fagin1994reasoning,belardinelli_interactions_2012}) which could be used to develop satisfiability solvers. However, there seems to be a lack of usable implementations. 

The component verification problem consists in showing that a component implementation $I$ correctly implements a given contract $C=(A,G)$. The verification methods may depend on the form used to express the implementation. When $I \in \mathcal{S}_{KP}$, then the problem can be reduced again to validity checking. The formula to be shown valid is $I\wedge A \models G$. However, in many cases, the implementation $I$ is directly given as code or a simulation model. Uncertainty quantification tools \cite{smith_uncertainty_2013,swiler_epistemic_2009} can be used in these cases. In the case of autonomous systems, specialized advanced tools are also available to find environments in which a specification is not met \cite{dreossi_compositional_2019,fremont_formal_2020}.

\section{Concluding remarks}
\label{sec:summary}
The set of models, methods, and tools for dealing with uncertainty in cyber-physical systems have advanced to a point where engineering can confidently deliver high-assurance systems. Cyber-physical systems have evolved into complex distributed, multi-agent systems that expose new challenges in analysis and design. Traditional techniques for dealing with uncertainty need to be complemented by other tools that can generate tight requirements under aleatoric and epistemic uncertainty. Reasoning compositionally about all these types of uncertainty is necessary to scale verification to large systems, to allow for integration of capabilities from multiple parties, and to enable incremental deployment of autonomous functions with minimal effort. 

We presented a methodology and a compositional modeling framework that addresses these new requirements, and we have shown their application to an unmanned aerial system. While this is a promising step, further research is needed in the area of automated reasoning for analysis and verification.

\section{Acknowledgement}
The NASA University
Leadership Initiative (grant \#80NSSC20M0163) provided funds to assist the authors with their research, but this article solely reflects the
opinions and conclusions of its authors and not any NASA entity.

\bibliographystyle{spmpsci}
\bibliography{citations}

\end{document}